\documentclass[leqno,11pt]{article}
\usepackage{amsmath,amsthm,amscd,amssymb}
\usepackage{titling}
\usepackage{enumerate}
\usepackage{latexsym}
\usepackage{graphicx,color}
\usepackage{psfrag}
\usepackage{fancyhdr}
\usepackage{fullpage}
\usepackage{enumitem}
\usepackage{tikz}
\usepackage[labelformat=simple]{subcaption}

\usetikzlibrary{bayesnet,arrows.meta}

\newtheorem{theorem}{Theorem}
\newtheorem{conj}{Conjecture}

\newtheorem{proposition}[theorem]{Proposition}
\newtheorem{corollary}[theorem]{Corollary}

\newtheorem{thm}{Theorem}[section]

\theoremstyle{definition}
\newtheorem{remark}[thm]{Remark}
\newtheorem{example}[thm]{Example}

\newtheorem{definition}[thm]{Definition}

\title{Revisiting a synthetic intracellular regulatory network that exhibits oscillations}

\author{Jonathan Tyler, Anne Shiu, Jay Walton}

\begin{document}
\maketitle

\begin{abstract}
In 2000, Elowitz and Leibler introduced the repressilator--a synthetic gene circuit with three genes that cyclically repress transcription of the next gene--as well as a corresponding mathematical model. Experimental data and model simulations exhibited oscillations in the protein concentrations across generations.  In 2006, M\"{u}ller \textit{et al.}\ generalized the model to an arbitrary number of genes and analyzed the resulting dynamics.  Their new model arose from five key assumptions, two of which are restrictive given current biological knowledge.  Accordingly, we propose a new repressilator system that allows for general functions to model transcription, degradation, and translation.  We prove that, with an odd number of genes, the new model has a unique steady state and the system converges to this steady state or to a periodic orbit.  We also give a necessary and sufficient condition for stability of steady states when the number of genes is even and conjecture a condition for stability for an odd number.  Finally, we derive a new rate function describing transcription that arises under more reasonable biological assumptions than the widely used single-step binding assumption.  With this new transcription-rate function, we compare the model's amplitude and period with that of a model with the conventional transcription-rate function.  Taken together, our results enhance our understanding of genetic regulation by repression.
\end{abstract}

\section{Introduction}
\label{intro}
The \textit{repressilator} is an experimental preparation used in synthetic biology to better understand genetic regulation by repression.  Introduced in 2000 by Elowitz and Leibler, the repressilator is a feedback loop consisting of three genes that each cyclically represses transcription of the next gene (Figure \ref{fig:repress_network}). The network was synthesized in \textit{E.coli} cells and exhibited sustained limit-cycle oscillations in single cells and across generations \cite{repressilator}.  

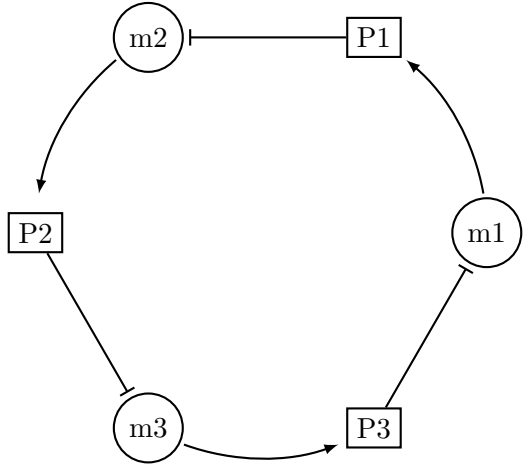
\begin{figure}[ht]
\centering
\begin{tikzpicture}

\def \n {5}
\def \radius {3cm}
\def \margin {10} 

\node[draw, circle, thick] (m1) at ({0}:\radius) {m1};
\draw[->, >=latex, thick] ({0+\margin}:\radius)
	arc ({0+\margin}:{360/6-\margin}:\radius);
\node[draw, rectangle, thick] (P1) at ({360/6}:\radius) {P1};
\node[draw, circle, thick] (m2) at ({360/6*2}:\radius) {m2};
\draw[->, >=latex, thick] ({360/6*2+\margin}:\radius)
	arc ({360/6*2+\margin}:{360/6*3-\margin}:\radius);
\node[draw, rectangle,thick] (P2) at ({360/6*3}:\radius) {P2};
\node[draw, circle, thick] (m3) at ({360/6*4}:\radius) {m3};
\draw[->, >=latex, thick] ({360/6*4+\margin}:\radius)
	arc ({360/6*4+\margin}:{360/6*5-\margin}:\radius);
\node[draw, rectangle,thick] (P3) at ({360/6*5}:\radius) {P3};
\edge[-Bar, shorten >= 2pt, thick, black] {P1}{m2};
\edge[-Bar, shorten >= 2pt, thick, black] {P2}{m3};
\edge[-Bar, shorten >= 2pt, thick, black] {P3}{m1};

\end{tikzpicture}
\caption{The repressilator network with three genes and their respective products \cite{repressilator}.  The $m$'s denote mRNA while the $P$'s denote proteins.  The product of gene 1 represses transcription of gene 2; the product of gene 2 represses transcription of gene 3; the product of gene 3 represses transcription of gene 1.}
\label{fig:repress_network}
\end{figure}

In addition to presenting experimental results, Elowitz and Leibler also introduced a mathematical model to describe the dynamics of the repressilator.  This model was subsequently generalized by M\"{u}ller \textit{et al.}\ in 2006 \cite{Muller2006}.  Specifically, M\"{u}ller \textit{et al.}\ analyzed two systems of ODEs that describe the dynamics of a repressilator with an arbitrary number of genes.  One system assumed that, in saturated amounts of repressors, transcription occurs at a very low rate.  Muller \textit{et al.}\ called this system \textit{RepLeaky} and proved results about the number of steady states, the stability of those steady states, and the limiting dynamics \cite{Muller2006}.  Here, the RepLeaky system is the starting point for our generalized repressilator model.

The RepLeaky system of M\"{u}ller \textit{et al.}\ arose from five key assumptions \cite{Muller2006}:
\begin{enumerate}[label = (\alph*)]
\item Genes are present in constant amounts.
\item When a protein binds to a regulatory element of a gene, it either enhances or inhibits transcription.  Also, binding reactions are in equilibrium.
\item Transcription and translation operate under saturated conditions.
\item Both mRNAs and free proteins are degraded by first-order reactions.
\item Transcription, translation, and degradation rates are the same among genes, mRNAs, and proteins, respectively.
\end{enumerate}
Two of these assumptions are biologically restrictive, so we generalize the model by removing them.  Consider, for example, the translation process.  In eukaryotic cells, mRNAs must be spliced correctly before they can exit the nucleus and then be translated \cite{splicing}. Similarly, since transcription depends on the uncoiling of DNA due to different locations of genes on histones \cite{uncoiling}, transcription rates should be allowed to vary across genes.  Finally, ubiquitization, which facilitates degradation, also differs extensively among proteins \cite{cell}.  Thus, to be more faithful to the biology, we remove assumption (e).

Next, we consider assumption (d).  Recently, Page and Perez-Carrasco have analyzed the repressilator after allowing for differing degradation rates among the proteins \cite{Page2018}.  Here, we argue for a further generalization.  In the context of the degradation pathway of a core clock component of the \textit{Neurospora} circadian clock, phosphorylation of the FREQUENCY (FRQ) protein initiates its own degradation.  This process occurs through the ubiquitin-proteasome pathway, which is a Michaelis-Menten pathway \cite{He953}.  Modeling the rate of FRQ degradation as proportional to its concentration is therefore not appropriate.  Thus, for our repressilator model, we remove assumption (d) to allow for more general functions than first-order terms.  In Section 2, we give conditions that these new terms must satisfy to reflect the biology of degradation.  We then prove results on how, if at all, these new terms change the dynamics of the model.

Finally, following the discussion in \cite{Kim2016}, we advocate for changing how we model repression and, in particular, we allow for a wider range of transcription-rate functions that satisfy a few biological assumptions.  The \textit{Hill function}, which is the standard transcription-rate function, arises from the following ``single-step assumptions" \cite{forger}:

\begin{enumerate}
\item On the promoter, either no repressor proteins are bound and transcription occurs, or repressors proteins are bound to all binding sites and no transcription occurs.
\item The repressor protein binds rapidly to the promoter.
\end{enumerate}
It is noted in \cite{forger}: ``As these assumptions are very restrictive, it is very surprising how often the Hill expression is used."  Accordingly, we introduce the following alternate set of assumptions, similar to those given in \cite{forger}:

\begin{enumerate}
\item There are $m$ binding sites on a promoter, and the repressor proteins bind in order from sites 1 to $m$.
\item Transcription cannot occur if $m$ repressor proteins are bound to the promoter.  Transcription can occur in all other cases.
\item The repressor protein binds rapidly to the promoter.
\item Repressor proteins bind to the $m$ binding sites at varying rates.
\end{enumerate}
We label these assumptions the \textit{successive-binding assumptions} and use them to derive a new transcription-rate function in Section 3.  
  
We prove that many of the results of M\"{u}ller \textit{et al.}\ extend to our generalized model of the repressilator.  First, with an odd number of genes, the system has a unique steady state, called the \textit{central steady state}, and the system converges to that steady state or produces limit-cycle oscillations.  Next, we prove a necessary and sufficient condition for stability of any steady state in the case of an even number of genes.  We also discuss what the condition means biologically.  In Section 3, we derive a new transcription-rate function from the successive-binding assumptions and show that it satisfies natural, biological conditions on models of transcription, presented in Section 2.  In Section 4, we numerically compare the amplitude and period of repressilator models constructed with the traditional transcription-rate function versus our newly derived function.  Finally, we end with a discussion in Section 5.

\section{General Repressilator System}

In this section, we introduce the new repressilator system and prove results about its steady states, stability, and asymptotic behavior.  First, we recall M\"uller \textit{et al.}'s RepLeaky model \cite{Muller2006}, which arises from a generalization of Figure \ref{fig:repress_network} to $n$ genes, and is given by the following system of $2n$ ODEs where $n$ denotes the number of genes:
\begin{equation}\label{eq:Muller_system}
\begin{aligned}
\dot{r}_i &= \alpha f(p_{i-1})-r_i, \\
\dot{p}_i &= \beta r_i-\beta p_i, 
\end{aligned}
\end{equation}
for $i=1,...,n$.  Here, $p_{i}$ denotes the concentration of protein-$i$, where $i$ is viewed mod $n$, and $r_i$ denotes the mRNA concentration.  The parameter $\beta$ is the ratio of protein degradation to mRNA degradation, and the parameter $\alpha$ is the transcription rate.  The function $f(x)$ models the repression of gene-$i$ transcription resulting from repressor protein-($i-1$) binding to the promoter (see Figure \ref{fig:repress_network}): 
\begin{equation}
f(x) = \frac{1-\delta}{1+x^h}+ \delta, \nonumber
\end{equation}
where the parameter $\delta$ is the ratio of repressed to unrepressed transcription \cite{Muller2006}.  Synthesis of protein-$i$ occurs by translation of mRNA-$i$ and is proportional to the mRNA-$i$ concentration.  Degradation of each species is modeled by a first-order term proportional to its own concentration.

As mentioned in the introduction, our aim is to generalize the repressilator by allowing for general degradation-rate, transcription-rate, and translation-rate functions as well as differing rate constants.  Our \textbf{generalized $n$-gene repressilator system}, which generalizes \eqref{eq:Muller_system}, is given by the following system of ODEs:
\begin{equation}\label{eq:repress}
\begin{aligned}
\dot{r}_1 &= a_1(p_n) - d_{r_1}(r_1),\\ 
\vdots \\ 
\dot{r}_n &= a_n(p_{n-1}) - d_{r_n}(r_n),\\ 
\dot{p}_1 &= k_1(r_1) - d_{p_1}(p_1), \\
\vdots \\
\dot{p}_n &= k_n(r_n) - d_{p_n}(p_n). \\
\end{aligned}
\end{equation}
Here, for the $i$-th gene, $r_i$ is the concentration of mRNA-$i$, and $p_i$ is the concentration of the protein.  Each equation in the system has a synthesis term and a degradation term.  One synthesis term is the function $a_i(p_{i-1})$, called the \textit{transcription-rate} function of gene-$i$ in terms of protein-($i-1$).  The degradation term for mRNA-$i$ is the \textit{degradation-rate} function $d_{r_i}(r_i)$, which is a function of its own concentration.  The function $k_i(r_i)$ is the \textit{translation-rate} function describing the synthesis of protein-$i$ in terms of mRNA-$i$.  Finally, the \textit{degradation-rate} function $d_{p_i}(p_i)$ models the degradation of protein-$i$ as a function of its own concentration.  

The 3-gene version of system \eqref{eq:repress} reflects Figure \ref{fig:repress_network}.  The \textbf{m1} node describes mRNA-1 which translates, according to the function $k_1(r_1)$, to protein-1, \textbf{P1}.  This protein then represses the synthesis of the second mRNA, which is described by the transcription-rate function $a_2(p_1)$.  

Next, we give conditions on the transcription-rate, degradation-rate, and translation-rate functions that we will assume for the results below.  These assumptions are rooted in the biology of the specific process they model.  For the transcription-rate functions, we begin with the biological assumptions.

\begin{description}
\item [(B1)] Transcription rates vary smoothly in the amount of repressor present.  
\item [(B2)] Transcription rates are always nonnegative.
\item [(B3)] Transcription rates decrease with increased repressor present.
\item [(B4)] Transcription rates are positive when no repressor is present.
\end{description}
These biological assumptions translate into the following mathematical assumptions on the transcription-rate function $a_i(x)$:
\begin{description}
\item [(A1)] $a_i(x) \in C^1[\mathbb{R}_{\geq 0}]$.
\item [(A2)] $a_i(\mathbb{R}_{\geq 0}) \subset \mathbb{R}_{\geq 0}$.
\item [(A3)] $a_i(x)$ is strictly decreasing on $\mathbb{R}_{\geq 0}$.
\item [(A4)] $a_i(0) > 0$.
\end{description}

The canonical transcription-rate function is $a_i(p) = \frac{k_i^S}{1+p^h}$ for some Hill coefficient $h$ \cite{Muller2006}.  This function is derived from the single-step binding assumptions listed in Section \ref{intro}, and it is easily seen that this function satisfies (A1)-(A4).  In Section 3, we derive another transcription-rate function using the successive-binding assumptions listed in Section \ref{intro} and show that this function also satisfies assumptions (A1)-(A4).

Next, we provide biological assumptions for degradation-rate and translation-rate functions.

\begin{description}
\item [(B1)] Degradation and translation rates vary smoothly in the protein or mRNA concentration.
\item[(B2)] Degradation and translation rates occur only when the protein or mRNA is present.
\item [(B3)] Degradation and translation rates increase as protein or mRNA concentrations increase.
\end{description}
These assumptions give rise to the following mathematical assumptions on the degradation-rate and translation-rate functions $d_{p_i}(x)$, $d_{r_i}(x)$, and $k_i(x)$.
\begin{description}
\item [(D1)] $d(x), k(x) \in C^1[\mathbb{R}_{\geq 0}]$.
\item[(D2)] $d(0)=k(0) = 0$. 
\item [(D3)] $d(x), k(x)$ are strictly increasing on $\mathbb{R}_{> 0}$.
\end{description} 

Notice immediately that degradation-rate and translation-rate functions satisfying (D1)-(D3) are invertible on their ranges.  This will be important in the following section.  

For the remainder of the paper, when considering our repressilator system~\eqref{eq:repress}, we assume that the functions $a_i(p_{i-1})$ satisfy (A1)-(A4), and the functions $d_{p_i}(p_i)$, $d_{r_i}(r_i)$, and $k_i(r_i)$ satisfy (D1)-(D3).  

\subsection{Steady States}
For system \eqref{eq:Muller_system}, M\"{u}ller \textit{et al.}\ proved the existence of a unique steady state, labeled $E_C$ for \textit{central steady state}, in the odd-$n$ case and also showed that $E_C$ exists in the even-$n$ case \cite{Muller2006}.  When we allow general transcription-rate and degradation-rate functions in system \eqref{eq:repress}, however, we are not always guaranteed a steady state.  Consider the following example.

\begin{example}\label{example:even_example} Consider the following $2$-gene version of the repressilator system \eqref{eq:repress}:
\begin{equation}\label{eq:ex1}
\begin{aligned}
\dot{r}_1 &= 2\pi - \arctan(p_2) - r_1 \\
\dot{r}_2 &= 2\pi - \arctan(p_1) - r_2\\
\dot{p}_1 &= r_1 - \arctan(p_1)\\
\dot{p}_2 &= r_2 - \arctan(p_2).
\end{aligned}
\end{equation}
\noindent It is straightforward to check that the assumptions (A1)-(A4) and (D1)-(D3) hold for the corresponding functions $a_i = 2\pi -\arctan(p_{i-1})$, $d_{r_i} = r_i$, and $d_{p_i} = \arctan(p_i)$.  We set the equations in \eqref{eq:ex1} to zero to solve for the steady states, giving
\begin{equation}\label{eq:eqn1}
2\pi - \arctan(p_2) = \arctan(p_1)
\end{equation}
\begin{equation}\label{eq:eqn2}
2\pi - \arctan(p_1) = \arctan(p_2).
\end{equation}
However, Eqns. \eqref{eq:eqn1} and \eqref{eq:eqn2} have no positive, real solution.  Therefore, system \eqref{eq:ex1} has no steady state.  The same is true if we augment system \eqref{example:even_example} to three genes using the same functions for the mRNA and protein, respectively.  
\end{example}
What went wrong in this example?  The degradation-rate function $d_{p_i}$ and the transcription-rate function $a_i$ each had a horizontal asymptote that prevented intersection of their respective graphs in $\mathbb{R}^2_+$.  This lack of intersection precluded the existence of a steady state.  So, to prove when steady states exist, we must introduce more assumptions.

Notice that assumptions (A2) and (A3) imply: 
\begin{equation}
\alpha_i:=\lim_{x \to \infty} a_i(x) < \infty \quad \text{and} \quad \lim_{x \to \infty} a_i^{\prime}(x) = 0. \nonumber
\end{equation}
This parameter $\alpha_i$ corresponds to the \textit{leakiness} of the promoter of gene-$i$ \cite{Muller2006}.  If $\alpha_i > 0$, then even in saturated amounts of repressor, gene-$i$ will still be transcribed at a positive rate, whereas $\alpha_i = 0$ implies that in saturated amounts of repressor, gene-$i$ will not be transcribed.  We introduce a new assumption on the transcription-rate function $a_i(p)$. 

\begin{description}
\item [(A5)] $\alpha_i = 0$ for all $i=1,...,n$.
\end{description}

Even if the leakiness $\alpha_i$ is nonzero, we can avoid the problem highlighted in Example \ref{example:even_example} by introducing an assumption on the relationship among the transcription-rate and degradation-rate functions.  Let us define 
\begin{equation}
\delta_i^R := \lim_{x \to \infty} d_{r_i}(x)  \quad \text{and} \quad   \delta_i^P := \lim_{x \to \infty} d_{p_i}(x).  \nonumber
\end{equation}
We allow for $\delta_i^R$ and $\delta_i^P$ to be infinite.  The $\delta_i^P$'s and $\delta_i^R$'s correspond to the maximum possible degradation rate for protein-$i$ and mRNA-$i$, respectively.  To avoid the problem in Example \ref{example:even_example}, we introduce the following relationship among $\delta_i^R$, $\delta_i^P$, and $a_i$: 

\begin{description}
\item [(A6)] $\delta_i^P > k_i(d_{r_i}^{-1}(a_i(0))) \quad \text{and} \quad \delta_i^R > a_i(0)$, for all $i=1,...,n$.
\end{description} 

Below, by using combinations of the above assumptions and others, we prove conditions under which $E_C$ exists, first with an odd number of genes, and then with an even number.

\subsubsection{Odd-$n$ Case}

For system \eqref{eq:Muller_system}, M\"{u}ller \textit{et al.}\ showed that the system has a unique steady state \cite{Muller2006}.  We prove that this property extends to system \eqref{eq:repress}.

\begin{proposition}\label{prop:unique}
For $n$ odd, if system \eqref{eq:repress} satisfies (A6), then system \eqref{eq:repress} has a unique steady state in $\mathbb{R}^{2n}_+$.
\end{proposition}

\begin{proof}
First, we set the equations in system \eqref{eq:repress} to zero:
\begin{equation}\label{eq:r}
0 = \dot{r}_i = a_i(p_{i-1}) - d_{r_i}(r_i) \implies d_{r_i}(r_i) = a_i(p_{i-1}) 
\end{equation}
\begin{equation}\label{eq:p}
0 = \dot{p}_i = k_i(r_i)-d_{p_i}(p_i) \implies d_{p_i}(p_i) = k_i(r_i). 
\end{equation}
From Eqns. \eqref{eq:r} and \eqref{eq:p}, it is easy to check that finding steady states reduces to finding solutions to the system
\begin{equation}
p_i = d_{p_i}^{-1}\circ k_i \circ d_{r_i}^{-1}\circ a_i(p_{i-1}) \quad \text{for } i = 1,\dots,n. \nonumber
\end{equation}
Write $f_i = d_{p_i}^{-1}\circ k_i \circ d_{r_i}^{-1}\circ a_i$, which, if assumption (A6) holds, is well defined.  

We compose the $f_i$'s to obtain a fixed-point problem:
\begin{equation}\label{eq:fixed}
p_i = f_i \circ f_{i-1} \circ \dots \circ f_1 \circ f_n \circ \dots \circ f_{i+1}(p_i), \quad \text{ for } i = 1,\dots,n.
\end{equation}
Since the $f_i$'s are monotonically decreasing by (A3) and (D3) and we are composing an odd number of functions, the composition in \eqref{eq:fixed} is monotonically decreasing.  It is also positive at 0 by (A3), (A4), (D2), and (D3).  Therefore, for $i=1,\dots,n$, there is exactly one solution to Eqn. \eqref{eq:fixed} in $\mathbb{R}^+$, so system \eqref{eq:repress} has a unique steady state in $\mathbb{R}^{2n}_+$.
\end{proof}
We follow the notation in \cite{Muller2006} and label this unique steady state as follows: 
\begin{definition}
The \textbf{central steady state}, $E_C$, is the concentration vector
\begin{equation}\label{eq:steady_state}
\left(d_{r_1}^{-1}\circ a_1(p_n^*), d_{r_2}^{-1}\circ a_2(p_1^*), \dots, d_{r_n}^{-1}\circ a_n(p_{n-1}^*), p_1^*, \dots, p_n^*\right), \
\end{equation}
where, (for $i=1,\dots,n$), $p_i^*$ solves Eqn.\ \eqref{eq:fixed}.
\end{definition}

\begin{remark}\label{rem:E_Css}
A solution to Eqn.\ \eqref{eq:fixed} gives a steady state as in \eqref{eq:steady_state} regardless of whether $n$ is even or odd because it solves a fixed-point problem derived from setting the equations of system \eqref{eq:repress} to zero. 
\end{remark}

\subsubsection{Even-$n$ Case}

Below, we give various conditions under which the fixed-point problem in Eqn. \eqref{eq:fixed} has a solution and consequently, guarantees when $E_C$ is a steady state. First, however, we must introduce another assumption on the degradation-rate functions.

\begin{description}
\item [(D4)] $(d_{p_i})^{\prime}(0) \neq 0$ and $(d_{r_i})^{\prime}(0) \neq 0$ for all $i=1,...,n$.
\end{description}

\begin{remark}
Assumption (D4) is biologically reasonable as many commonly used degradation-rate functions satisfy (D4), e.g., linear degradation and Michaelis-Menten kinetics.  However, there exist degradation processes that do not satisfy (D4).  For example, consider a protein that is selected for degradation by dimerization with itself.  If we model this scenario with a quadratic degradation term, then it will not satisfy assumption (D4).    
\end{remark}

With assumption (D4), we can now prove various conditions under which system \eqref{eq:repress} admits a steady state.  

\begin{proposition}\label{prop:even_ss}
For system \eqref{eq:repress} with $n$ even, if assumptions (A5), (A6), and (D4) hold, then $E_C$ exists and is a steady state.
\end{proposition}

\begin{proof}
We follow the notation used in Proposition \ref{prop:unique} and show that there exists a solution to the fixed-point problem from \eqref{eq:fixed}:
\begin{equation}\label{eq:comp1}
p_i = f_i \circ f_{i-1} \circ \dots \circ f_1 \circ f_n \circ \dots \circ f_{i+1}(p_i). 
\end{equation}
Note that all $f_i$'s in Eqn.\ \eqref{eq:comp1} are well defined by assumption (A6).

In Eqn.\ \eqref{eq:comp1}, we are composing an even number of strictly decreasing functions, so the composition is strictly increasing. We also know that the composition is positive at zero by (A2), (A3), (D2), and (D3).  We will show that 
\begin{equation}
\lim_{x \to \infty} (f_i \circ f_{i-1} \circ \dots \circ f_1 \circ f_n \circ \dots \circ f_{i+1})^{\prime}(x) = 0. \nonumber
\end{equation}
This, along with the composition being positive at zero, will imply that $E_C$ exists. We compute:
\begin{equation}
(f_i \circ f_{i-1} \circ \dots \circ f_1 \circ f_n \circ \dots \circ f_{i+1})^{\prime}  \nonumber
\end{equation}
\begin{equation}
= (f_i^{\prime} \circ f_{i-1} \circ \dots \circ f_1 \circ f_n \circ \dots \circ f_{i+1})(f_{i-1}^{\prime} \circ f_{i-2} \dots \circ f_1 \circ f_n \circ \dots \circ f_{i+1}) \nonumber
\end{equation}
\begin{equation}
(f_{i-2}^{\prime} \circ f_{i-3} \dots \circ f_1 \circ f_n \circ \dots \circ f_{i+1})\cdots f_{i+1}^{\prime}. \nonumber
\end{equation}
First, we show that $\lim_{x \to \infty} f_{i+1}^{\prime}(x) = 0$.  The following calculations are straightforward and follow from (A5), (D2), and (D4):
\begin{align}
f_{i+1}^{\prime}(x) =  ((d_{p_{i+1}}^{-1})^{\prime} \circ k_{i+1} \circ d_{r_{i+1}}^{-1}\circ &a_{i+1}(x))\cdot \nonumber\\ (k_{i+1}^{\prime}\circ d_{r_{i+1}}^{-1}\circ a_{i+1}(x)) &\cdot((d_{r_{i+1}}^{-1})^{\prime} \circ a_{i+1}(x))\cdot a_{i+1}^{\prime}(x),  \label{eq:8}\\ \nonumber \\
\lim_{x \to \infty} (d_{r_{i+1}}^{-1})^{\prime} \circ a_{i+1}(x) &= \lim_{x \to 0} (d_{r_{i+1}}^{-1})^{\prime} (x)\nonumber \\  &= \frac{1}{(d_{r_{i+1}})^{\prime}(0)} < \infty,  \label{eq:9}\\ \nonumber \\
\lim_{x \to \infty} (k_{i+1}^{\prime}\circ d_{r_{i+1}}^{-1}\circ a_{i+1}(x)) &= k_{i+1}^{\prime}(0) < \infty, \\ \nonumber \\
\lim_{x \to \infty} (d_{p_{i+1}}^{-1})^{\prime}\circ k_{i+1} \circ d_{r_{i+1}}^{-1}\circ a_{i+1}(x) &= \lim_{x \to 0} (d_{p_{i+1}}^{-1})^{\prime} (x)\nonumber \\ &= \frac{1}{(d_{p_{i+1}})^{\prime}(0)} < \infty.  \label{eq:10}
\end{align}
It is easy to check that Eqns.\ \eqref{eq:8}-\eqref{eq:10} imply: 
\begin{equation}
\lim_{x \to \infty} f_{i+1}^{\prime}(x) = 0. \nonumber
\end{equation}
Now we show that for $k = i,...,1,n,...,i+2$: 
\begin{equation}
\lim_{x \to \infty} (f_k^{\prime} \circ f_{k-1} \circ \dots \circ f_{i+1})(x) < \infty. \nonumber
\end{equation}
Recall that $f_i = d_{p_i}^{-1} \circ k_i \circ d_{r_i}^{-1}\circ a_i$.  Then, by (A5), 
\begin{equation}
\lim_{x \to \infty} (f_k^{\prime} \circ f_{k-1} \dots \circ f_{i+1})(x) = (f_k^{\prime} \circ f_{k-1} \dots \circ d_{p_{i+1}}^{-1} \circ k_{i+1} \circ d_{r_{i+1}}^{-1})(0) < \infty. \nonumber
\end{equation}
Therefore, 
\begin{align}
&\lim_{x \to \infty} (f_i \circ f_{i-1} \circ \dots \circ f_1 \circ f_n \circ \dots \circ f_{i+1})^{\prime}(x) \nonumber \\
= &\lim_{x \to \infty} (f_i^{\prime} \circ f_{i-1} \circ \dots \circ f_1 \circ f_n \circ \dots \circ f_{i+1}) \cdot \nonumber \\ &\lim_{x \to \infty} (f_{i-1}^{\prime} \circ f_{i-2} \dots \circ f_1 \circ f_n \circ \dots \circ f_{i+1})
\cdots \lim_{x \to \infty} f^{\prime}_{i+1}(x) = 0. \nonumber
\end{align}
Since $i$ was arbitrary, each $p_i$ has a solution, and $E_C$ exists and by Remark \ref{rem:E_Css} is a steady state. 
\end{proof}
   
\begin{proposition}\label{prop:even_ss2}
Consider system \eqref{eq:repress} with $n$ even.  If $\alpha_i>0$ for all $i=1,...,n$ and (A6) holds, then $E_C$ exists and is a steady state.
\end{proposition}
\begin{proof}
The proof is similar to the proof of Proposition \ref{prop:even_ss}.  Assumption (A6) implies that the inverses of $d_{r_i}(r_i)$ and $d_{p_i}(p_i)$ exist at $a_i(0)$ for all $i$.  Also, by assuming that $\alpha_i > 0$, both $$\lim_{x \to \infty} (d_{r_i}^{-1})^{\prime} \circ a_{i}(x) \quad \text{and} \quad \lim_{x \to \infty} (d_{p_i}^{-1})^{\prime} \circ k_i \circ d_{r_i}^{-1}\circ a_{i}(x)$$ are finite because $d_{p_i}^{\prime}(\alpha_i), d_{r_i}^{\prime}(\alpha_i) > 0$ by assumption (D3).
\end{proof}

We present a final sufficient condition for when $E_C$ is a steady state in the even-$n$ case.  The condition is motivated by the following example. 

\begin{example}\label{example:even_example2}
Consider the following generalized 2-gene repressilator model:
\begin{equation}
\begin{aligned}
\dot{r}_1 &= \frac{1}{1+p_{2}^2} - r_1^2 \\
\dot{r}_2 &= \frac{1}{1+p_1^2} - r_2^2\\
\dot{p}_1 &= r_1 - p_1^2 \\
\dot{p}_2 &= r_2 - p_2^2. \nonumber
\end{aligned}
\end{equation}
\end{example}
This model fails the assumptions of Proposition \ref{prop:even_ss}, namely (D4), because the derivatives of the degradation-rate functions $d_{p_i} = p_i^2$ at zero are zero, and it fails those of Proposition~\ref{prop:even_ss2} because $\alpha_1 = \alpha_2 = 0$.  Nevertheless, $E_C$ exists and is a steady state, because $E_C$ is the solution to the following system:
\begin{equation}
\begin{aligned}
p_1^4 = \frac{1}{1+p_2^2} \\
p_2^4 = \frac{1}{1+p_1^2}. \nonumber
\end{aligned}
\end{equation}
Finding the fixed point is equivalent to solving: 
\begin{equation}\label{eq:even_motivation}
p^4 = \frac{1}{1+p^2}. 
\end{equation}
The left-hand side of Eqn. \eqref{eq:even_motivation} is zero at zero and increases to $\infty$ while the right-hand side is greater than zero at zero and decreasing, so $E_C$ exists.  This phenomenon leads to our final result about $E_C$ in the even-$n$ case.

\begin{proposition}\label{prop:even_ss3}
Consider system \eqref{eq:repress} with $n$ even or odd.  Assume that all the degradation-rate functions $d_{r_i}$ are equal ($=d_r$), all the degradation-rate functions $d_{p_i}$ are equal ($=d_p$), all the transcription-rate functions $a_i$ are equal ($=a$), and all the translational-rate functions $k_i$ are equal ($=k$).  If $\lim_{x \to \infty} k(x) > \delta^P$, where $\delta^P := \lim_{x \to \infty} d_{p}(x)$, then $E_C$ exists and is a steady state.
\end{proposition}
\begin{proof}
Under the assumptions of the proposition, it is easy to check that computing $E_C$ reduces to solving
\begin{equation}\label{eq:evenprop}
a(p) = d_r\circ k^{-1} \circ d_p(p)
\end{equation}
for $p \in \mathbb{R}^+$.  The composition $d_r \circ k^{-1} \circ d_p(p)$ is well defined for all $p > 0$ by the assumption that $\lim_{x \to \infty} k(x) > \delta^P$.  Also, the function $a(p)$ is decreasing, while the composition $d_r \circ k^{-1} \circ d_p(p)$ is increasing.  Finally, $a(0) > d_r \circ k^{-1} \circ d_p(0) = 0$ by assumptions (A4) and (D2).  Therefore, there is a solution $p \in \mathbb{R}^+$ to Eqn.\ \eqref{eq:evenprop}, so $E_C$ exists. 
\end{proof}

\begin{remark}
The combinations of assumptions in Propositions \ref{prop:unique}-\ref{prop:even_ss3} used to prove existence of $E_C$ provide insight into possible repressilator design circuits.  For example, a design circuit with a low-copy plasmid and proteins that are signaled for degradation through dimerization with itself could be problematic because the system may not have a steady state.  Likewise, assumption (A6)--used in the proofs of Propositions \ref{prop:unique}-\ref{prop:even_ss2}--requires that the maximal mRNA degradation rate ``overcome" the maximal transcription rate.  We revisit the theme of comparing degradation rates and synthesis rates when we address the stability of steady states in the next section.
\end{remark}

\subsection{Stability Analysis}

For their model, M\"{u}ller \textit{et al.}\ proved general results about the stability of the central steady state by harnessing the fact that the matrix $J - \lambda I$, where $J$ is the Jacobian of system \eqref{eq:repress} at $E_C$, is a circulant matrix.  This matrix representation allowed the eigenvalues to be represented in terms of roots of unity, which in turn allowed for identifying general inequalities in the parameters that characterize stability.  For our generalized repressilator model, however, the matrix $J-\lambda I$ does not reduce to a circulant matrix.  Thus, we use different methods to characterize stability.

We begin with a few definitions.

\begin{definition}\label{def:DK}
Consider the generalized repressilator system \eqref{eq:repress}.  Let $x^* \in \mathbb{R}_+^{2n}$.
\begin{enumerate}
\item The \textbf{$i$-th mRNA degradation rate} at $x^*$ is
\begin{equation}
\partial_i^R = \frac{\mathrm{d}d_{r_i}(r_i)}{\mathrm{d}r_i}\Bigr |_{x^*}. \nonumber
\end{equation}
\item The \textbf{$i$-th protein degradation rate} at $x^*$ is
\begin{equation}
    \partial_i^P = \frac{\mathrm{d}d_{p_i}(p_i)}{\mathrm{d}p_i}\Bigr |_{x^*}. \nonumber
\end{equation}
\item The \textbf{$i$-th degradation product} at $x^*$ is
\begin{equation}
\mathcal{D}_i := \partial_i^R \partial_i^P. \nonumber
\end{equation}
\item The \textbf{total degradation product} at $x^*$ is 
\begin{equation}
\mathcal{D} := \prod_{i=1}^n \mathcal{D}_i, \nonumber
\end{equation}
where $\mathcal{D}_i$ is the $i$-th degradation product at $x^*$.
\item The \textbf{$i$-th synthesis product} at $x^*$ is \begin{equation}
\mathcal{K}_i := \left(\frac{\mathrm{d}k_i(r_i)}{\mathrm{d}r_i}\Bigr |_{x^*}\right) \left(\frac{\mathrm{d}a_i(p_{i-1})}{\mathrm{d}p_{i-1}}\Bigr |_{x^*}\right). \nonumber
\end{equation}
\item The \textbf{total synthesis product} at $x^*$ is 
\begin{equation}
\mathcal{K} := \prod_{i=1}^n \mathcal{K}_i, \nonumber
\end{equation}
where $\mathcal{K}_i$ is the $i$-th synthesis product at $x^*$.
\end{enumerate}
\end{definition}
When $n$ is even, the total synthesis product is positive, because the even number of repression elements in the cycle results in what Mallet-Paret and Smith call a \textit{positive feedback system} \cite{Mallet-Paret1990}.  In the odd-$n$ case, the total synthesis product is negative, because the system is a \textit{negative feedback system}.  These differences play an important role in determining the stability of $E_C$.

Throughout the section, we will refer to the Routh-Hurwitz criterion, so we review it briefly.  Consider a univariate polynomial:
\begin{equation}\label{eq:polynomial}
p(x) = a_n + a_{n-1}x + a_{n-2}x^2 + ... + a_0x^n.
\end{equation}
\begin{definition}\label{def:hurwitz}
For $k = 1,...n$, the \textbf{$k^{th}$ Hurwitz matrix} of $p$ as in \eqref{eq:polynomial} is the $k \times k$ matrix $H_k = [h_{ij}]_{i,j=1}^k$, defined by $h_{ij} = a_{2i-j}$, where $a_{2i-j}$ is defined as 0 if $2i-j<0$ or $2i-j>n$.  
\end{definition}
For example, the fourth Hurwitz matrix of $p(x) = a_4 + a_3x + a_2x^2 + a_1x^3+a_0x^4$ is:
\[
  H_4 =
  \left[ {\begin{array}{cccc}
   a_1 & a_0 & 0 & 0 \cr a_3 & a_2& a_1 &a_0\cr 0 & a_4 & a_3 & a_2 \cr 0 & 0 & 0 & a_4\  \end{array} } \right].
\]
Following the notation in \cite{yang}, we write $D_i = \det(H_i)$. 

\begin{theorem}[Routh-Hurwitz Criterion \cite{allen}]\label{thm:R-H}
Consider a polynomial $p$ as in \eqref{eq:polynomial}.  Every root of $p$ has negative real part if and only if the determinants of all Hurwitz matrices (Definition \ref{def:hurwitz}) are positive, i.e., 
\begin{equation}
D_i > 0, \quad i = 1,2,...,n. \nonumber
\end{equation}
\end{theorem}
Recall that the stability of a steady state is characterized by negative real parts of the roots of the characteristic polynomial of the Jacobian.  Thus, we apply Theorem \ref{thm:R-H} to this characteristic polynomial to obtain a necessary and sufficient condition for the stability of a steady state (see Theorems \ref{thm:even_stability} and \ref{thm:det_Hurs}).

\subsubsection{Even-$n$ Case}
For system \eqref{eq:Muller_system} with $n$ even, M\"{u}ller \textit{et al.}\ found a condition on the derivative of the transcription-rate function that characterizes when the central steady state is stable.  Here, we generalize that criterion to system \eqref{eq:repress} using $\mathcal{D}$ and $\mathcal{K}$.

\begin{theorem}\label{thm:even_stability}
Consider system \eqref{eq:repress} with $n$ even.   A steady state $x^*$ is locally asymptotically stable if and only if
\begin{equation}\label{eq:inequality}
\mathcal{D} > \mathcal{K}, 
\end{equation}
where $\mathcal{D}$ and $\mathcal{K}$ are evaluated at $x^*$.  
\end{theorem}

\begin{proof} 
\begin{figure}
\centering
\begin{tikzpicture}[scale=0.50]
\draw[->, thick] (-5,0)--(5,0) node[right]{$Re$};
\draw[->, thick] (0,-5)--(0,5) node[above]{$Im$};
\draw[ultra thick, ->] (0,-3) arc (-90:45:3) node[right]{$\gamma_1$};
\draw[ultra thick] (0,-3) arc (-90:90:3);
\draw[ultra thick, ->] (0,3) -- (0,1/2) node[left]{$\gamma_2$}; 
\draw[ultra thick] (0,1/2) -- (0,-3);
\node[below right=1pt] at (3,0){$R$};
\node[left=1pt] at (0,3){$Ri$};
\node[left=1pt] at (0,-3){$-Ri$};
\end{tikzpicture}
\caption{Contour $\Gamma$ in proof of Theorem \ref{thm:even_stability}.}
\label{Gamma}
\end{figure}
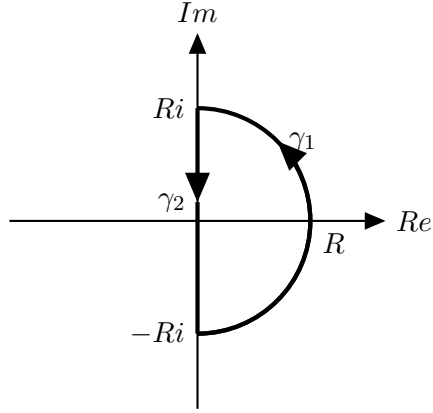
It is easily checked that the characteristic polynomial of the Jacobian matrix of system \eqref{eq:repress} at $x^*$ is
\begin{equation*}
p(\lambda) = \prod_{j=1}^{n}(\lambda+\partial_j^R)(\lambda+\partial_j^P) - \mathcal{K}. 
\end{equation*}
It follows that the constant term of $p$ is $\mathcal{D}-\mathcal{K}$. 

($\implies$)We use the Routh-Hurwitz criterion.  Assume that system is stable at $x^*$.  Then $\det(H_{n-1}) > 0$ and $\det(H_n) > 0$.  However, $\det(H_n) = \det(H_{n-1}) \cdot (\mathcal{D}-\mathcal{K})$ implying that $\mathcal{D}-\mathcal{K} > 0$, i.e., $\mathcal{D}>{\mathcal{K}}$.  

$(\impliedby)$ We use Rouch\'e's Theorem \cite{complex}.  Write $p_1(z) = \prod_{j=1}^{n}(z+\partial_j^R)(z+\partial_j^P)$ and $p_2(z) = \mathcal{K}$.  We will show that the number of zeros of $p(\lambda)$ in the right-hand half plane is equal to the number of zeros of $p_1$ in the right-hand half plane.  Since all $\partial_j$'s are positive by assumption (D3), there are no zeros of $p_1(z)$ in the right-hand half plane, so there are no zeros of $p(\lambda)$.  

Consider the contour described by the semicircle of radius $R$ in the right-hand half plane along with the line segment connecting $-Ri$ and $Ri$ on the imaginary axis.  Call the contour $\Gamma$ (Figure \ref{Gamma}).  We separate $\Gamma$ into the semicircle, $\gamma_1$, and the line, $\gamma_2$.  This is a closed contour in the complex plane.  First, we show that $|p_1(z)| > |p_2(z)|$ on $\gamma_1$.  We write $z = Re^{i\theta}$ on $\gamma_1$.  Then 
\begin{equation*}
|p_1(z)| = |p_1(Re^{i\theta})| = \prod_{j=1}^n|Re^{i\theta} + \partial_j^R||Re^{i\theta} + \partial_j^P| \geq \prod_{j=1}^n|R - \partial_j^R||R - \partial_j^P|,  
\end{equation*}
by the reverse triangle inequality.  Call $d$ the maximum of the degradation constants.  Then
\begin{equation*}
\prod_{j=1}^n|R - \partial_j^R||R - \partial_j^P| \geq \prod_{j=1}^{2n}|R - d|. \nonumber
\end{equation*}
Let $R^{\prime} = 2d+1$. Then for all $R \geq R^{\prime}$, 
\begin{equation*}
\prod_{j=1}^{2n}|R - d| > d^{2n} \geq \mathcal{D}.
\end{equation*}
Therefore, for contours $\Gamma$ with a sufficiently large radius, by assumption \eqref{eq:inequality}, the following inequalities hold on $\gamma_1$:
\begin{equation}
|p_1(z)| > \mathcal{D} > |\mathcal{K}| = |p_2(z)|. \nonumber
\end{equation}
Now all that is left to show is that $|p_1(z)|>|p_2(z)|$ on $\gamma_2$.  On $\gamma_2$, we write $z = iy$ for $-R<y<R$.  Then
\begin{equation}
|p_1(z)| = \prod_{j=1}^n|iy+\partial_j^R||iy+\partial_j^P| \geq \prod_{j=1}^n|Re(iy+\partial_j^R)||Re(iy+\partial_j^P)| = \mathcal{D}. \nonumber
\end{equation}
Therefore, again by assumption \eqref{eq:inequality}, the following holds on $\gamma_2$:
\begin{equation}
|p_1(z)| > \mathcal{D} > \mathcal{K} = |p_2(z)|. \nonumber
\end{equation}
The number of zeros of $p_1(z) + p_2(z)$ in $\Gamma$ is the same as the number of zeros of $p_1(z)$ in $\Gamma$ for all $R \geq R^{\prime}$.  Since $\partial_i^R, \partial_i^P> 0$ for all $i$, we know that there are no zeros of $p(\lambda)$ inside $\Gamma$ for all $R\geq R^{\prime}$.  Therefore, there are no eigenvalues of the Jacobian with positive or zero real part, so the system is stable. 
\end{proof}

Theorem \ref{thm:even_stability} has the following biological interpretation.  Inequality \eqref{eq:inequality} says that, in the long term, degradation is a more powerful process than synthesis.  Thus, system \eqref{eq:repress} converges locally if and only if degradation is stronger than the combined synthesis of mRNA and protein.

\subsubsection{Odd-$n$ Case}

Recall that, in Proposition \ref{prop:unique}, we proved $E_C$ always exists and is unique when $n$ is odd. Below, we prove results towards finding a necessary and sufficient condition for stability of $E_C$ in the odd-$n$ case like we have in the even case from Theorem \ref{thm:even_stability}.  Our proofs use Hurwitz matrices because the inherent structure of the system when $n$ is odd allows us to simplify the Routh-Hurwitz criterion.  Towards the end of the section, we conjecture a necessary and sufficient condition for stability of $E_C$ and then give evidence for it.

First, we discuss why the proof of Theorem \ref{thm:even_stability} does not generalize to the odd-$n$ case.  Recall that, in this case, system \eqref{eq:repress} is a negative feedback loop and $\mathcal{K} < 0$, while in the even case, $\mathcal{K} > 0$.  Thus, in the odd case, $\mathcal{D} > \mathcal{K}$ always holds, not only when the system is stable.  Also, even though $\mathcal{D}>0>\mathcal{K}$, we are not guaranteed that 
\begin{equation}\label{eq:odd_ineq}
\mathcal{D} > |\mathcal{K}|, 
\end{equation}
which is what we used in the proof of Theorem \ref{thm:even_stability}.  If inequality \eqref{eq:odd_ineq} does hold, however, we conclude that the system is stable at $E_C$.

\begin{proposition}\label{prop:odd_stab}
Consider system \eqref{eq:repress} with $n$ odd.  If inequality \eqref{eq:odd_ineq} holds, then $E_C$ is locally asymptotically stable.
\end{proposition}
\begin{proof}
The proof is the same as in the backwards direction of Theorem \ref{thm:even_stability}.
\end{proof}

We continue to solve the question of stability at $E_C$ by using the structure of the system to reduce the number of Hurwitz matrices needed in the Routh-Hurwitz criterion.  The idea is that the characteristic polynomial of the system is close to a polynomial that is known to have all negative real roots and so we will need to check fewer Hurwitz determinants.

\begin{theorem}\label{thm:det_Hurs}
Consider system \eqref{eq:repress} with $n$ odd, and let $D_i$ denote the determinant of the $i$-th Hurwitz matrix of the Jacobian at $E_C$.  Then $E_C$ is locally stable if and only if $D_i > 0 $ for all $i=n+2,\dots,2n-1$.  
\end{theorem}

\begin{proof}
We first show that, when $n$ is odd, the first $n+1$ Hurwitz matrices calculated from the characteristic polynomial of the Jacobian at $E_C$ always have positive determinant.

Recall from the proof of Theorem \ref{thm:even_stability} that the characteristic polynomial of the Jacobian matrix at $E_C$ is $p(\lambda) = \prod_{i=1}^n (\lambda+\partial_i^R)(\lambda+\partial_i^P) - \mathcal{K}$, where $\mathcal{K}$ is the total synthesis product from Definition \ref{def:DK}.  Since $n$ is odd and so $\mathcal{K} < 0$, we rewrite this as $p(\lambda) = \prod_{i=1}^n (\lambda+\partial_i^R)(\lambda+\partial_i^P) + |\mathcal{K}|$.  We introduce a new polynomial $q(\lambda) = \prod_{i=1}^n (\lambda+\partial_i^R)(\lambda+\partial_i^P)$.  

In what follows, any quantity with a superscript $p$ is constructed using $p(\lambda)$, and similarly for $q(\lambda)$.  Notice that $p(\lambda)$ and $q(\lambda)$ both have degree $2n$, so there are $2n$ Hurwitz matrices $H_i^p$ for $p(\lambda)$ and $H_i^q$ for $q(\lambda)$.  Also, all coefficients of $p(\lambda)$ and $q(\lambda)$ match except for the constant term.  Therefore, every Hurwitz matrix constructed using only coefficients of $p(\lambda)$ that are not the constant term is equivalent to the corresponding Hurwitz matrix of $q(\lambda)$.  We will use this fact below.

We now split the proof into two cases.
\begin{enumerate}
\item Case 1: $i = 1,..., n$.

From Definition \ref{def:hurwitz}, the coefficients of the polynomial that appear in $H_i$ are indexed by $1,..., 2i-1$.  Therefore, $H_i^p = H_i^q$ for $i=1,...,n$, so $d_{p_i} > 0$ for $i=1,...,n$ because all roots of $q(\lambda)$ have negative real part.

\item Case 2: $i = n+1$. 

For this case, we examine the effect of the constant term of $p$ on the determinant of $H^p_{n+1}$. Recall that $a_{2n}^p = a_{2n}^q + |\mathcal{K}|$ where $a_{2n}^p$ and $a_{2n}^q$ are the constant terms of $p$ and $q$, respectively.  

Below, we use $A^{[a,b]}$ to denote the matrix $A$ without row-$a$ and column-$b$.  The $(n+1)$st Hurwitz matrix of $p$ is the following $(n+1) \times (n+1)$ matrix:

\[
  H_{n+1}^p =
  \left[ {\begin{array}{cccccc}
   a_1 & a_0 & 0 & 0&\dots & 0 \cr a_3 & a_2& a_1 &a_0 & \dots & 0\cr a_5 & a_4 & a_3 & a_2 & a_1 & a_0 \cr \vdots & \vdots & \vdots & \vdots & \vdots & \vdots \cr 0& a^p_{2n} & a_{2n-1} & \dots & a_{n+2} & a_{n+1}\  \end{array} } \right],
\]
and $H^q_{n+1}$ matches $H^p_{n+1}$ at all entries except for entry $(n+1,2)$, where it is the constant term $a_{2n}^q$ rather than that of $p$.  We compute $D_{n+1}^p = \det(H_{n+1}^p)$ and $D_{n+1}^q = \det(H_{n+1}^q)$ by expanding along the last row:
\begin{align}
D_{n+1}^p &= a_{2n}^p\det(H_{n+1}^{p,[n+1,2]})-a_{2n-1}\det(H_{n+1}^{p,[n+1,3]})+\dots \nonumber \\
&+ a_{n+1}\det(H_{n+1}^{p,[n+1,n+1]}),  \label{eq:Dp}
\end{align}
and
\begin{align}
D_{n+1}^q &= a_{2n}^q\det(H_{n+1}^{q,[n+1,2]})-a_{2n-1}\det(H_{n+1}^{q,[n+1,3]})+\dots \nonumber \\
&+ a_{n+1}\det(H_{n+1}^{q,[n+1,n+1]}).  \label{eq:Dq}
\end{align}
As the constant term is present only in the last row of $H_{n+1}$, the submatrices of $H_{n+1}^p$ and $H_{n+1}^q$ that exclude that row are equal.  Combining this fact with Eqns. \eqref{eq:Dp} and \eqref{eq:Dq} gives
\begin{equation}\label{eq:DpminusDq}
\begin{aligned}
D_{n+1}^p - D_{n+1}^q &= a_{2n}^p\det(H_{n+1}^{p,[n+1,2]}) - a_{2n}^q\det(H_{n+1}^{q,[n+1,2]}) \\
&= (a_{2n}^p-a_{2n}^q)\det(H_{n+1}^{p,[n+1,2]}) = |\mathcal{K}|\det(H_{n+1}^{p,[n+1,2]}).
\end{aligned}
\end{equation}

To compute the determinant of the following matrix:
\[
  H^{p,[n+1,2]}_{n+1} =
  \left[ {\begin{array}{ccccc}
   a_1  & 0 & 0&\dots & 0 \cr a_3 & a_1 &a_0 & \dots & 0\cr a_5 & a_3 & a_2 & a_1 & a_0 \cr \vdots  & \vdots & \vdots & \vdots & \vdots \cr a_{2n-1}& a_{2n-3} & a_{2n-4} & \dots & a_{n-1} \  \end{array} } \right]
\]
we expand about the first row, so $\det(H_{n+1}^{p,[n+1,2]}) = a_1\det(A)$, where 

\[
  A =
  \left[ {\begin{array}{cccccccc}
   a_1  & a_0 & 0 & 0 & 0&0&\dots & 0 \cr a_3 & a_2 &a_1 & a_0 & 0 &0&\dots & 0\cr a_5 & a_4 & a_3 & a_2 & a_1 & a_0 & \dots & 0\cr \vdots  & \vdots & \vdots & \vdots & \ddots &\vdots & \ddots & \vdots \cr a_{2n-3}& a_{2n-4} & a_{2n-5} & a_{2n-6} & \dots  & a_{n+1}&a_{n} & a_{n-1} \  \end{array} } \right].
\]
Notice that $A = H_{n-1}^p = H_{n-1}^q$, which has positive determinant by Case 1.  Therefore, because $a_1 >0$, $\det(H_{n+1}^{p,[n+1,2]}) = a_1det(A) > 0$.  Since all roots of $q(\lambda)$ have negative real part, $D_{n+1}^q > 0$ by Theorem \ref{thm:R-H}, so Eqn. \eqref{eq:DpminusDq} gives
\begin{equation}
D_{n+1}^p - D_{n+1}^q = |\mathcal{K}|\det(H_{n+1}^{p,[n+1,2]}) > 0 \implies D_{n+1}^p > D_{n+1}^q > 0. \nonumber 
\end{equation}

\end{enumerate}
Therefore, $D_{n+1}^p > 0$ and so the first $n+1$ determinants of the Hurwitz matrices constructed from $p(\lambda)$ are positive.  

Since $D_{2n} = (\mathcal{D} - \mathcal{K})D_{2n-1}$ and $\mathcal{D} - \mathcal{K} > 0$ (as explained above Proposition~\ref{prop:odd_stab}), we conclude from Theorem \ref{thm:R-H} that $E_C$ is locally stable if and only if $D_i > 0 $ for all $i=n+2,\dots,2n-1$. 
\end{proof}

\begin{corollary}\label{cor:Cor_1}
For $n=3$, system \eqref{eq:repress} is stable at $E_C$ if and only if $D_5 > 0$.  
\end{corollary}

\begin{proof}
Follows immediately from Theorem \ref{thm:det_Hurs}.
\end{proof}

Next, we recall the stability condition for $E_C$ due to M\"uller \textit{et al.}\ and compare it to the one in Theorem \ref{thm:det_Hurs}.  M\"uller \textit{et al.}'s\ criterion \cite{Muller2006} is:

\begin{equation}\label{eq:Muller_stable}
\frac{\beta}{(1+\beta)^2} < \frac{1-S_c\cos(\pi/n)}{S_c^2\sin^2(\pi/n)},  
\end{equation}
where
\begin{equation}
S_c = -\alpha f^{\prime}(E_C).  
\end{equation}
In system \eqref{eq:repress}, it is easy to see that $S_c$ equals $-\mathcal{K}_i$ at $E_C$.  Therefore, we rewrite Eqn.\ \eqref{eq:Muller_stable} as:
\begin{equation}\label{adj_eq:Muller_stable}
\frac{\beta}{(1+\beta)^2} < \frac{1+\mathcal{K}_i\cos(\pi/n)}{\mathcal{K}_i^2\sin^2(\pi/n)}.
\end{equation}

For $n=3$, it is straightforward to check that inequality \eqref{adj_eq:Muller_stable} is equivalent to:

\begin{equation}\label{eq:rearranged1}
(4+2\mathcal{K}_i)(1+\beta)^2-3\beta\mathcal{K}_i^2 > 0. 
\end{equation}

For system \eqref{eq:Muller_system} with $n=3$, by Corollary~\ref{cor:Cor_1}, the condition $D_5>0$ characterizes the same stability region in parameter space as inequality~\eqref{eq:rearranged1}.  This is surprising because $D_5$ under system \eqref{eq:Muller_system} and $n=3$ is a more complicated expression than the left-hand side in \eqref{eq:rearranged1}:

\begin{align}\label{eq:g}
D_5 &=\beta^2(8\beta^{10}\mathcal{K}_i^3+64\beta^{10}+144\beta^9\mathcal{K}_i^3 + 576\beta^9 +792\beta^8\mathcal{K}_i^3+2304\beta^8-27\beta^7\mathcal{K}_i^6 \nonumber \\
&+2184\beta^7\mathcal{K}_i^3+5376\beta^7-81\beta^6\mathcal{K}_i^6 
+3528\beta^6\mathcal{K}_i^3+8064\beta^6-81\beta^5\mathcal{K}_i^6 \nonumber\\
&+3528\beta^5\mathcal{K}_i^3+8064\beta^5-27\beta^4\mathcal{K}_i^6+2184\beta^4\mathcal{K}_i^3 +5376\beta^4+792\beta^3\mathcal{K}_i^3   \nonumber \\
&+2304\beta^3+144\beta^2\mathcal{K}_i^3+576\beta^2+8\beta\mathcal{K}_i^3+64\beta). 
\end{align}

Next, we prove directly that these two inequalities define the same stability region when $\beta \in \mathbb{R}_{>0}$ and $\mathcal{K}_i \in \mathbb{R}$.  Note that, by definition, $\mathcal{K}_i$ is always negative, but we show that even for $\mathcal{K}_i \in \mathbb{R}$ the two inequalities are equivalent. 

\begin{theorem}[Equivalence of the $n=3$ stability conditions]
For $n=3$ of system \eqref{eq:Muller_system}, inequality \eqref{eq:rearranged1} holds for $\beta \in \mathbb{R}_{>0}$ and $\mathcal{K}_i \in \mathbb{R}$ if and only if $D_5>0$, where $D_5$ is the determinant of the Hurwitz matrix $H_5$ of the characteristic polynomial of the Jacobian matrix of \eqref{eq:Muller_system} evaluated at $E_C$. 
\end{theorem}

\begin{proof}
Let $f(\beta, \mathcal{K}_i) = (4+2\mathcal{K}_i)(1+\beta)^2-3\beta\mathcal{K}_i^2$ denote the polynomial on the left-hand side of \eqref{eq:rearranged1}.  We rename $D_5$, as in \eqref{eq:g}, the polynomial $g(\beta, \mathcal{K}_i)$.   We must show that $f(\beta, \mathcal{K}_i)$ and $g(\beta, \mathcal{K}_i)$ are the same sign for all $\beta \in \mathbb{R}_{>0}$ and $\mathcal{K}_i \in \mathbb{R}_{<0}$.  

It is straightforward to check, e.g. using $\tt{Maple}$, that $g(\beta, \mathcal{K}_i) = f(\beta, \mathcal{K}_i)h(\beta, \mathcal{K}_i)$, where
\begin{equation}\label{h}
\begin{aligned}
h(\beta, \mathcal{K}_i) &= 9\mathcal{K}_i^4\beta^8 + 6\mathcal{K}_i^3\beta^9+4\mathcal{K}_i^2\beta^{10} +27\mathcal{K}_i^4\beta^7+30\mathcal{K}_i^3\beta^8+40\mathcal{K}_i^2\beta^9 \\
&-8\mathcal{K}_i\beta^{10}+27\mathcal{K}_i^4\beta^6+60\mathcal{K}_i^3\beta^7+144\mathcal{K}_i^2\beta^8-56\mathcal{K}_i\beta^9+16\beta^{10}\\
&+9\mathcal{K}_i^4\beta^5+60\mathcal{K}_i^3\beta^6+260\mathcal{K}_i^2\beta^7-168\mathcal{K}_i\beta^8 +112\beta^9 +30\mathcal{K}_i^3\beta^5\\
&+260\mathcal{K}_i^2\beta^6-280\mathcal{K}_i\beta^7+336\beta^8+6\mathcal{K}_i^3\beta^4 +144\mathcal{K}_i^2\beta^5-280\mathcal{K}_i\beta^6 \\
&+560\beta^7+40\mathcal{K}_i^2\beta^4-168\mathcal{K}_i\beta^4+560\beta^6+4\mathcal{K}_i^2\beta^3-56\mathcal{K}_i\beta^4+336\beta^5\\
&-8\mathcal{K}_i\beta^3+112\beta^4+16\beta^3. \nonumber
\end{aligned}
\end{equation} 
Because $g =fh$, any root of $f$ is also a root of $g$.  We will use this fact below.

Fix $\tilde{\beta}>0$.  Let $g_{\tilde{\beta}}(\mathcal{K}_i) := g(\tilde{\beta},\mathcal{K}_i) $ and $f_{\tilde{\beta}}(\mathcal{K}_i):= f(\tilde{\beta},\mathcal{K}_i)$.  We rewrite $g_{\tilde{\beta}}$:
\begin{align}
g_{\tilde{\beta}}(\mathcal{K}_i) &= \mathcal{K}_i^6(-81\tilde{\beta}^6-81\tilde{\beta}^5-27\tilde{\beta}^4)+\mathcal{K}_i^3(8\tilde{\beta}^{10}+144\tilde{\beta}^9 \nonumber\\
&+792\tilde{\beta}^8+2184\tilde{\beta}^7+3528\tilde{\beta}^6+3528\tilde{\beta}^5+2184\tilde{\beta}^4\nonumber \\
&+792\tilde{\beta}^3+144\tilde{\beta}^2+8\tilde{\beta}) + C, \label{eq:rearr_g}
\end{align}
where $C$ is the sum of all the pure $\beta$ terms in \eqref{eq:g}.  It is easy to check that $C> 0$ when $\tilde{\beta}>0$.  Thus, we see from \eqref{eq:rearr_g} that the polynomial $g_{\tilde{\beta}}$ has one sign change.  Therefore, by Descartes' rule of signs, $g_{\tilde{\beta}}$ has at most one positive real root and at most one negative real root.  

From \eqref{eq:rearranged1}, $f_{\tilde{\beta}}(\mathcal{K}_i)$ is a quadratic polynomial in $\mathcal{K}_i$ that is downward facing and has a positive y-intercept namely, (4$(1+\tilde{\beta})^2$).  Therefore, $f_{\tilde{\beta}}$ has exactly two real roots, and thus, $g_{\tilde{\beta}}$ has exactly two real roots as well because $g=fh$ and, as noted above, $g_{\tilde{\beta}}$ has at most two real roots.  

We label these two real roots $r_1$ and $r_2$ with $r_1 < r_2$.  Since $g_{\tilde{\beta}}$ has even degree in $\mathcal{K}_i$ with a negative leading coefficient and a positive y-intercept, we know that $g_{\tilde{\beta}}>0$ if and only if $\mathcal{K}_i$ is in the interval $(r_1,r_2)$.  It is straightforward to check that $f_{\tilde{\beta}}(\mathcal{K}_i)$ also is positive if and only if $\mathcal{K}_i$ is in the interval $(r_1,r_2)$.  Therefore, $f_{\tilde{\beta}}>0$ if and only if $g_{\tilde{\beta}}>0$.  Our choice of $\tilde{\beta}>0$ was arbitrary.  Therefore, the two inequalities $D_5>0$ and \eqref{adj_eq:Muller_stable} are equivalent.  
\end{proof}

Corollary \ref{cor:Cor_1} and the fact that M\"uller \textit{et al.}'s\ criterion for system \eqref{eq:Muller_system} is given by a single inequality lead us to conjecture that, when $n$ is odd, stability of $E_C$ depends only on the penultimate Hurwitz determinant.  

\begin{conj}\label{conj}
For $n$ odd, system \eqref{eq:repress} is stable at $E_C$ if and only if $D_{2n-1} > 0$.
\end{conj}

Evidence for Conjecture \ref{conj} can be seen in the possible types of bifurcations of $E_C$ in the odd case.  We reorder the species as $r_1, p_1, r_2, p_2,...$ to see that system \eqref{eq:repress} is a \textit{monotone system}--a system that satisfies $\dot{x}_i = f(x_i, x_{i-1})$ for all $i$.  In \cite{Mallet-Paret1990}, Mallet-Paret and Smith showed that all omega-limit sets of monotone systems can be embedded in $\mathbb{R}^2$.  Therefore, the possible bifurcations are stationary bifurcations or simple Hopf bifurcations.  However, there cannot be stationary bifurcations because zero is never a root of the characteristic polynomial.  Therefore, all bifurcations are simple Hopf bifurcations.  Furthermore, from \cite{yang}, at simple Hopf bifurcations, the following conditions hold: $D_1$, ..., $D_{2n-2} > 0$, and $D_{2n-1} = D_{2n} = 0$.  This reasoning is not sufficient to prove the conjecture, however, because there could be a point in parameter space where $E_C$ is unstable but nevertheless $D_{2n-1} > 0$.  

Finally, we prove a result about the global dynamics of system \eqref{eq:repress}, which is similar to Theorem 2 in \cite{Muller2006}, by using the result on monotone systems given in \cite{Mallet-Paret1990}.

\begin{theorem}\label{thm:periodic}
For $n$ odd, system \eqref{eq:repress} has the following properties: (i) Every orbit converges to $E_C$ or to a periodic orbit.  (ii) If $E_C$ is unstable, then there exists a periodic-orbit attractor.  
\end{theorem}

\begin{proof}
It is straightforward to check that the proof is the same as that of Theorem 2 in \cite{Muller2006}, which uses \cite{Mallet-Paret1990}.  We note that we can rule out the third option of the Main Theorem in \cite{Mallet-Paret1990} because $E_C$ is unique, so there are no heteroclinic or homoclinic orbits. 
\end{proof}

Theorem \ref{thm:periodic} is significant biologically because it shows the species concentrations of the repressilator constructed with an odd number of genes will either stabilize to the steady state value or to a limit-cycle.

\section{Transcription-Rate Function from Successive-Binding}

In \cite{Muller2006}, M\"{u}ller \textit{et al.}\ used a function arising from the single-step binding assumption, discussed above in Section 1, to model the binding of a gene product repressor to the next gene's promoter.  Here, we derive a new function to model binding of the gene product and promoter based on the successive-binding reaction mechanism, and then use it to define a new transcription-rate function.

First, we recall, from \cite{Muller2006}, the function that models the amount of binding as a function of the gene product and the promoter, $c_i^{(m)}$, and the resulting transcription-rate function, $a_i$:
\begin{equation}\label{eq:trad_ai}
c_i^{(m)} = \bar{g}\frac{p_{i-1}^m}{K + p_{i-1}^m} \quad \text{and} \quad a_i = \bar{g}[(1-\delta)(1-s(\frac{p_{i-1}}{K})) + \delta] ,
\end{equation}
where $\bar{g}$ is the total gene concentration; $\delta$ is the ratio of repressed to unrepressed transcription; $K$ is a dissociation constant; and
\begin{equation}\label{eq:old_transcrate}
s(x) = \frac{x^h}{1+x^h}, 
\end{equation}
where the Hill coefficient, $h > 0$.  One advantage to using the transcription-rate function \eqref{eq:old_transcrate} from the single-step binding assumption is that it generalizes naturally with any positive, real Hill coefficient.  

\subsection{Successive-Binding Function}

Next, we recall the assumptions for successive-binding introduced in Section~1.
\begin{enumerate}
\item There are $m$ binding sites on a promoter, and the repressor proteins bind in order from sites 1 to $m$.
\item Transcription cannot occur if $m$ repressor proteins are bound to the promoter.  Transcription can occur in all other cases.
\item The repressor protein binds rapidly to the promoter.
\item Repressor proteins bind to the $m$ binding sites at varying rates.
\end{enumerate}
These assumptions are adapted from \cite[Chapter 2]{forger} where Forger presents three models of repression.  The model we are interested in is his Model ``a'': A Model for Transcription Regulation with Independent Binding Sites.  

Here, we present the reaction mechanism and follow the notation in \cite{Muller2006}.  Let $\textbf{G}_i$ be gene-$i$; and $\textbf{P}_{i-1}$ the repressor produced by the preceding gene.  We write the \textit{gene-repressor complex} as $\textbf{C}_i^{(m)}$.  The \textit{successive-binding reaction mechanism}~is
\begin{align}
\textbf{G}_i+\textbf{P}_{i-1} &\rightleftharpoons \textbf{C}_i^{(1)} \nonumber \\
\textbf{C}_i^{(1)} + \textbf{P}_{i-1} &\rightleftharpoons \textbf{C}_i^{(2)} \nonumber \\
&\vdots  \label{eq:suc_bind} \\
\textbf{C}_i^{(m-1)} + \textbf{P}_{i-1} &\rightleftharpoons \textbf{C}_i^{(m)}. \nonumber
\end{align}

Assumption 1 presumes that the promoter has $m$ binding sites and that repressors bind in order from site 1 to $m$, so the mechanism has $m$ possible gene-repressor complexes $\textbf{C}_i^{(1)}, ..., \textbf{C}_i^{(m)}$.  We will derive the binding function $c_i^{(m)}$ that models the amount of binding as a function of the total gene concentration and concentration of the repressor present.  We proceed with this derivation below. 

Assumption 3 allows us to use the \textit{quasi steady state assumption} on the concentrations of the gene-repressor complexes to derive the binding function.  The binding function for $\textbf{C}_i^{(1)}$ is 
\begin{equation}\label{eq:c_i}
c_i^{(1)} = \frac{g_ip_{i-1}}{K_1}, 
\end{equation}
where $K_1$ is a dissociation constant.  Here, dissociation constants for each gene are distinct because of Assumption 4.  We use the function \eqref{eq:c_i} to write the binding function for $\textbf{C}_i^{(2)}$:
\begin{equation}
c_i^{(2)} = \frac{p_{i-1}c_i^{(1)}}{K_2} = \frac{g_ip_{i-1}^2}{K_1K_2}, \nonumber
\end{equation}
where $K_2$ is another dissociation constant.  We continue this process to get a general formula for the binding function of the $j$-th complex:
\begin{equation}\label{eq:general}
c_i^{(j)} = \frac{g_ip_{i-1}^j}{K_1K_2\cdots K_j},
\end{equation}
where $K_1,\dots,K_j$ are all dissociation constants.  

Conservation of mass for genes is given by
\begin{equation}\label{eq:conservation}
\bar{g} = g_i + c_i^{(1)} + c_i^{(2)} + \dots + c_i^{(m)}. 
\end{equation}
This conservation equation differs from the conservation equation arising from single-step binding.  Under single-step binding, the genes are either free or consumed in the final gene-repressor complex, leading to the conservation equation:
\begin{equation}
\bar{g} = g_i + c_i^{(m)}.  \nonumber
\end{equation}

We desire a binding function that depends only on the protein product concentration and the total gene concentration.  To obtain such a function, we must first solve for $c_i^{(m)}$ in terms of $p_{i-1}$ using Eqns.\ \eqref{eq:general} and \eqref{eq:conservation}.
\begin{equation}
c_i^{(m)} = \frac{(\bar{g} - c^{(1)}_i - \cdots - c^{(m-1)}_i - c^{(m)}_i)p_{i-1}^m}{K_1K_2 \cdots K_m} \nonumber
\end{equation}
\begin{equation}
\implies c_i^{(m)} = \frac{\bar{g}p_{i-1}^m}{K_1 K_2 \cdots K_m} - \frac{c_i^{(m)}p_i}{K_1} - \cdots - \frac{c_i^{(m)}p_{i-1}^{m-1}}{K_1K_2\cdots K_{m-1}} - \frac{c^{(m)}p_{i-1}^m}{K_1K_2\cdots K_m} \nonumber
\end{equation}
\begin{equation}
\implies c^{(m)}\left(1 + \frac{p_{i-1}}{K_1} + \frac{p_{i-1}^2}{K_1K_2} + \dots + \frac{p_{i-1}^m}{K_1K_2 \cdots K_m}\right) = \frac{\bar{g}p_{i-1}^m}{K_1K_2 \cdots K_m} \nonumber
\end{equation}
\begin{equation}
\implies c^{(m)}\left(\frac{K_1K_2 \cdots K_m + K_2\cdots K_m p_{i-1} + \dots + K_{m-1}p_{i-1}^{m-1} + p_{i-1}^m}{K_1K_2 \cdots K_m}\right) \nonumber
\end{equation}
\begin{equation}
= \frac{\bar{g}p_{i-1}^m}{K_1 K_2 \cdots K_m} \nonumber
\end{equation}
\begin{equation}
\implies c^{(m)} = \frac{\bar{g}p_{i-1}^m}{\sum_{j=0}^m ((\prod_{l> j}K_l) p_{i-1}^j)}.  \nonumber
\end{equation}
Similarly, we obtain $c_i^{(j)}$
\begin{equation}\label{eq:old_binding}
c^{(j)}_i = \frac{(\prod_{\ell>j}K_\ell)\bar{g}p_{i-1}^j}{\sum_{j=0}^m ((\prod_{\ell> j}^m K_\ell) p_{i-1}^j)}. 
\end{equation}
We simplify notation by letting $B_i(p_{i-1}) = \sum_{j=0}^m ((\prod_{\ell> j}K_\ell) p_{i-1}^j)$ and $A_i^{(j)}(p_{i-1}) = (\prod_{\ell>j}^m K_\ell)p_{i-1}^j$, so that:
\begin{equation}\label{AB_relation}
B_i(p_{i-1}) = \prod_{j=1}^m K_j + \sum_{j=1}^m A_i^{(j)}(p_{i-1}).
\end{equation}
Therefore, we rewrite Eqn.\ \eqref{eq:old_binding}, the \textit{successive-binding function}, as
\begin{equation}\label{eq:binding}
\boxed{c^{(j)}_i = \frac{\bar{g}A_i^{(j)}(p_{i-1})}{B_i(p_{i-1})} \text{   for   } j=1,\dots, m.} 
\end{equation}

\subsection{Transcription-rate Function Obtained from Successive-binding Function}

We assume as in \cite{Muller2006} that the transcription rate $a_i$ depends linearly on the free gene concentration $g_i$ given by the two cases
\begin{equation}\label{eq:case1}
g_i = \bar{g} \implies a_i = \bar{g}, 
\end{equation}
and
\begin{equation}\label{eq:case2}
g_i = 0 \implies a_i = \delta \bar{g}. 
\end{equation}
Here, following M\"{u}ller \textit{et al.} \cite{Muller2006}, $\delta$ denotes the ratio of repressed to unrepressed transcription.  Case \eqref{eq:case1} assumes that, if the gene is free of any repressors, then transcriptional activity will occur proportional to the total gene concentration.  Case \eqref{eq:case2} assumes that, if $m$ repressors are bound to the gene, then transcriptional activity will occur proportional to the constant $\delta$. 

From cases \eqref{eq:case1} and \eqref{eq:case2}, the transcription-rate $a_i$ is given by
\begin{equation}
a_i = (1-\delta)g_i+\delta \bar{g}. \nonumber
\end{equation}
We use Eqns. \eqref{eq:conservation} and \eqref{eq:binding} to rewrite $a_i$:
\begin{equation}\label{eq:pre_a_i}
a_i = \bar{g}\left[(1-\delta)\left(1-\frac{A_i^{(1)}(p_{i-1})+A_i^{(2)}(p_{i-1}) + \dots + A_i^{(m)}(p_{i-1})}{B_i(p_{i-1})}\right) + \delta\right]. 
\end{equation}
Using Eqn. \eqref{AB_relation}, we rewrite Eqn. \eqref{eq:pre_a_i} as
\begin{equation}
a_i = \bar{g}\left[\frac{(1-\delta)\prod_{j=1}^mK_j}{B_i(p_{i-1})} + \delta\right]. \nonumber
\end{equation}
To simplify notation, let us write
\begin{equation}\label{eq:final_S_i}
S_i(p_{i-1}) := \frac{\prod_{j=1}^mK_j}{B_i(p_{i-1})}. 
\end{equation}
Then, from Eqns. \eqref{eq:pre_a_i} and \eqref{eq:final_S_i}, the derived transcription-rate function is:
\begin{equation}\label{eq:final_a_i}
\boxed{a_i = \bar{g}[(1-\delta)S_i(p_{i-1}) + \delta].} 
\end{equation}
It is straightforward to check that Eqn. \eqref{eq:final_a_i} satisfies assumptions (A1)-(A4), and hence is a valid transcription-rate function.

\begin{proposition}\label{prop:6}
The transcription-rate function arising from the successive-binding mechanism, given by Eqn. \eqref{eq:final_a_i}, satisfies assumptions (A1)-(A4).
\end{proposition}
Propositions \ref{prop:unique} and \ref{prop:6} immediately yield the following corollary.

\begin{corollary}
Consider system \eqref{eq:repress} with $n$ odd and transcription-rate functions $a_i(p_{i-1})$ Eqn. \eqref{eq:final_a_i}, that is, arising from the successive-binding mechanism.  Then the central steady state $E_C$ exists and is the unique, positive steady state.
\end{corollary}

\begin{remark}
Forger, in \cite{forger}, simplifies Eqn.\ \eqref{eq:final_S_i} by assuming that the dissociation constant, $K_j$, is the same across each reaction in the successive-binding mechanism \eqref{eq:suc_bind}. Hence, his version of Eqn. \eqref{eq:final_S_i} is: 
\begin{equation}\label{eq:Forger_S_i}
S(p_{i-1}) = \frac{K^m}{(K+p_{i-1})^m}.  \nonumber
\end{equation}
\end{remark}

 \section{Comparison of Models Arising from Hill Functions vs. Successive-Binding Transcription-Rate Functions}

Below, we numerically compare a model using the traditional single-step binding assumption for transcription and another model constructed using the successive-binding assumption.  Specifically, we show that the amplitudes and periods of the oscillations can differ widely (see Figures \ref{fig:Amp} and \ref{fig:Per}).  

The first model is the following three-gene repressilator system:
\begin{align}\label{eq:model1}
    \dot{r}_1 &= \frac{k_1}{1+p_3^h} - r_1, & \dot{p}_1 &= 4r_1 - 3p_1 \nonumber\\
    \dot{r}_2 &= \frac{k_2}{1+p_1^h} - r_2, & \dot{p}_2 &= r_2 - 2p_2 \tag{SS}\\
    \dot{r}_3 &= \frac{k_3}{1+p_2^h} - r_3, &  \dot{p}_3 &= 4r_3 - 4p_3.\nonumber
\end{align}
Model \eqref{eq:model1} (for single-step) is constructed using the single-step binding assumption for each transcription-rate function, and the Hill coefficients, $h$, are assumed to be equal.  
  
In comparison, the second model considered is:
\begin{align}\label{eq:model2}
    \dot{r}_1 &= \frac{k_1}{(1+p_3)^h} - r_1, &  \dot{p}_1 &= 4r_1 - 3p_1\nonumber\\
    \dot{r}_2 &= \frac{k_2}{(1+p_1)^h} - r_2, &  \dot{p}_2 &= r_2 - 2p_2 \tag{SB}\\
    \dot{r}_3 &= \frac{k_3}{(1+p_2)^h} - r_3, & \dot{p}_3 &= 4r_3 - 4p_3. \nonumber 
\end{align}
Model \eqref{eq:model2} (for successive-binding) is constructed using the successive-binding assumption for each transcription-rate function (Eqn.\ \eqref{eq:final_a_i}), and, like model \eqref{eq:model1}, the Hill coefficients, $h$, are assumed to be equal.  Note from systems \eqref{eq:model1} and \eqref{eq:model2} that the two models are equivalent in the degradation and translation components. 

\subsection{Amplitude}

Here, we compare the amplitudes of models \eqref{eq:model1} and \eqref{eq:model2}. For the first numerical comparison, we vary the Hill coefficient, $h$, from 1 to 10 while keeping all other parameters fixed.  For both models \eqref{eq:model1} and \eqref{eq:model2}, we numerically solve the system until it reaches a steady state or a limit cycle.  Then, we compute the amplitude of protein 1 by evaluating the difference of the maximum and minimum protein 1 concentration.  Figure \ref{fig:amp_diff} shows the amplitudes of the first protein concentration with respect to the Hill coefficient (sampled at every one-tenth value--1, 1.1, 1.2, etc.) for models \eqref{eq:model1} (blue) and \eqref{eq:model2} (red).  All computations were performed in $\tt{MATLAB}$ \cite{MatlabOTB}. 

As shown in Figure \ref{fig:amp_diff}, the amplitude of model \eqref{eq:model1} increases to an order of magnitude larger than the amplitude of model \eqref{eq:model2}.  Also, the Hopf bifurcation of model \eqref{eq:model1} with respect to the Hill coefficient occurs when $h \approx 2$ whereas the Hopf bifurcation of model \eqref{eq:model2} happens when $h \approx 3$.  Thus, numerical evidence suggests that the traditional transcription-rate function allows for oscillations to occur at smaller Hill coefficients than for our newly derived transcription-rate function.  This means, in terms of the biology, that under the single-step binding assumption, oscillations can occur when there are fewer repressors binding to the gene promoter.  However, incorporating intermediate steps into the repressor-promoter interactions (like in the successive-binding assumption) leads to more repressors required to produce oscillations. 

Next, we conducted a numerical comparison that fixed all parameters ($h=3$) while letting the transcription rates, $k_1$, $k_2$, and $k_3$, vary.  In order to plot the amplitudes, we assume that $k_1 = k_2 = k_3 = k$ and let $k$ vary from 1 to 10.  Again, we sample $k$ at every one-tenth interval and numerically solve both models to convergence to the steady state or the limit cycle.  We then compute the amplitudes as in the first comparison.  Figure \ref{fig:trans_diff} shows the amplitudes of the first protein concentration with respect to the transcription rate for both models.  

Similar to the first comparison, model \eqref{eq:model1} amplitudes are significantly different from those of model \eqref{eq:model2}, and in fact, they reach an order of magnitude difference (Figure \ref{fig:trans_diff}).  Moreover, the Hopf bifurcation of model \eqref{eq:model1} occurs when $k \approx 2$ while the Hopf bifurcation of model \eqref{eq:model2} happens when $k \approx 4$ (Figure \ref{fig:trans_diff}).  

\begin{figure}
    \centering
    \begin{subfigure}{.49\textwidth}
    \includegraphics[width = \textwidth]{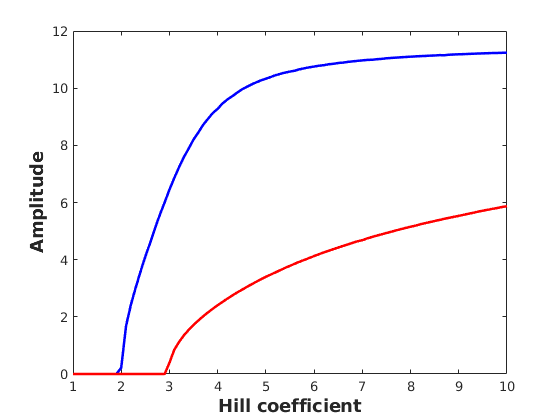}
    \caption{}
    \label{fig:amp_diff}
    \end{subfigure}
    \begin{subfigure}{.49\textwidth}
    \includegraphics[width = \textwidth]{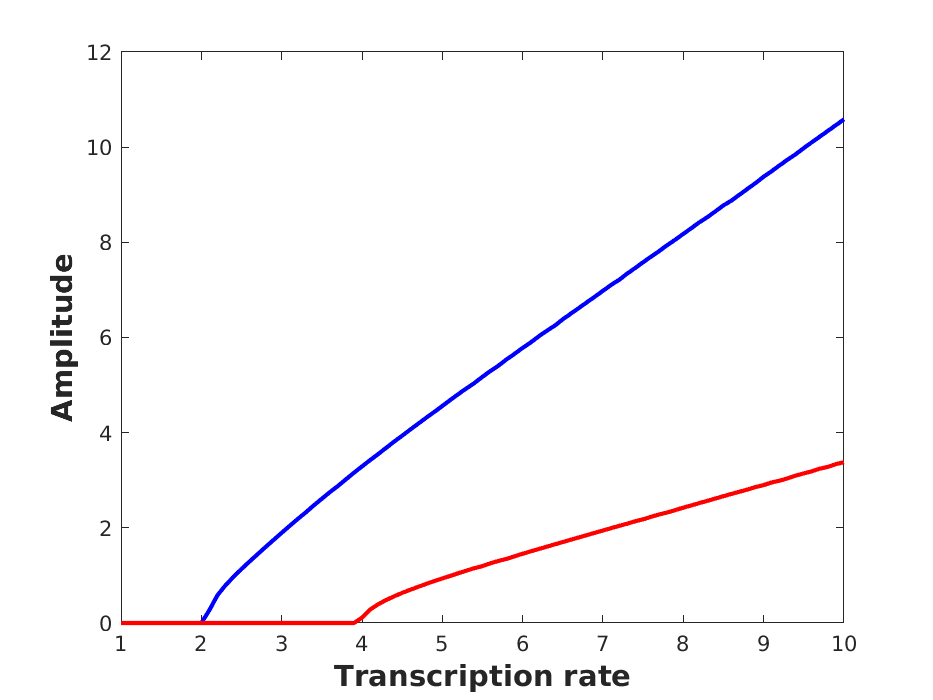}
    \caption{}
    \label{fig:trans_diff}
    \end{subfigure}
    \caption{(a) Amplitudes of the concentration of protein 1 for models \eqref{eq:model1} (blue curve) and \eqref{eq:model2} (red curve) with respect to the Hill coefficient.  We fixed the transcription rates as follows: $k_1 = 10$, $k_2 = 7$, $k_3=9$.  (b) Amplitudes of the concentration of protein 1 for models \eqref{eq:model1} (blue) and \eqref{eq:model2} (red) with respect to the transcription rate.  The initial conditions for both (a) and (b) were $r_1 = 10$, $r_2 = 2$, $r_3 = 3$, $p_1 = 5$, $p_2 = 1$, and $p_3=6$.}
    \label{fig:Amp}
\end{figure}

\begin{figure}
    \centering
    \begin{subfigure}{.49\textwidth}
    \includegraphics[width = \textwidth]{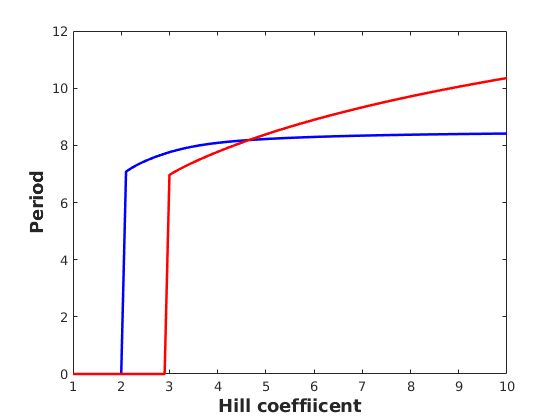}
    \caption{}
    \label{fig:per_hill_diff}
    \end{subfigure}
    \begin{subfigure}{.49\textwidth}
    \includegraphics[width = \textwidth]{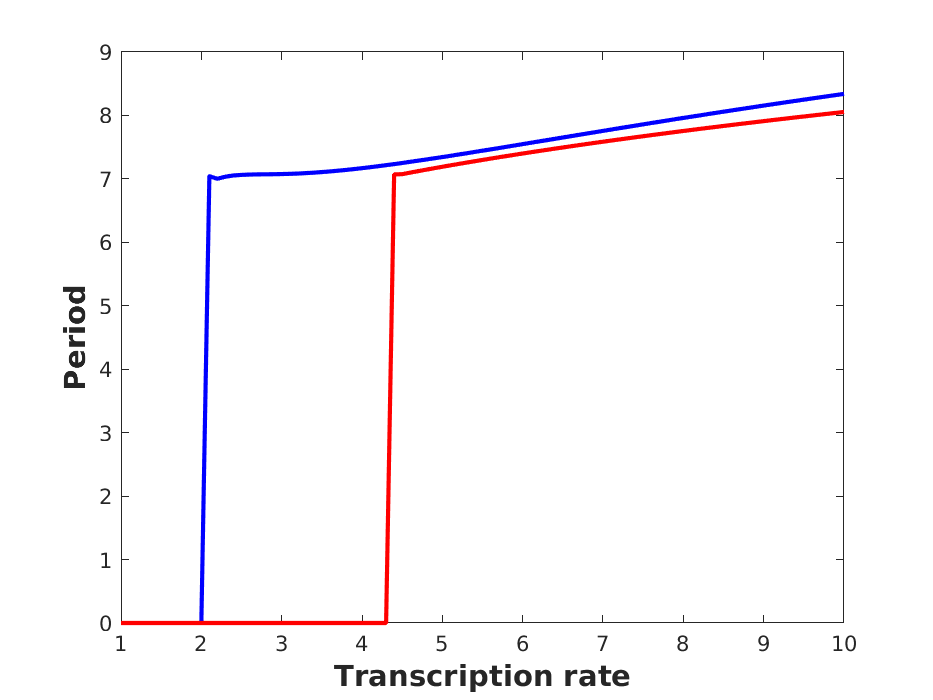}
    \caption{}
    \label{fig:per_tr_diff}
    \end{subfigure}
    \caption{(a) Period of the concentration of protein 1 for models \eqref{eq:model1} (blue) and \eqref{eq:model2} (red) with respect to the Hill coefficient.  Similar to the amplitude comparison in Figure \ref{fig:amp_diff}, the parameters $k_1$, $k_2$, and $k_3$ were set to 10, 7, and 9, respectively.  (b) Period of the concentration of protein 1 for models \eqref{eq:model1} (blue) and \eqref{eq:model2} (red) with respect to the transcription rate.  The Hill coefficient, $h$, was fixed at 4 for the simulations.  The initial conditions for both (a) and (b) were $r_1 = 10$, $r_2 = 2$, $r_3 = 3$, $p_1 = 5$, $p_2 = 1$, and $p_3=6$.}
    \label{fig:Per}
\end{figure}

These and other numerical simulations support the claim that the amplitude of a model constructed using the successive-binding assumption will be smaller than the amplitude of a model constructed using the single-step binding assumption, all other components being equal.  As amplitudes are an important quantity of oscillations of a system, care should therefore be taken when considering appropriate models of genetic repression and transcription or when fitting models to actual repressilator data.  

\subsection{Period}

Similar to the amplitude, the two transcription-rate functions yield dramatic differences in periods.  To compare, we compute the periods of models~\eqref{eq:model1} and \eqref{eq:model2}, again using $\tt{MATLAB}$.  First, we fix all parameters except the Hill coefficient, $h$.  Again, we let $h$ vary from 1 to 10 and sample $h$ at every one-tenth value.  We numerically solve the systems to either the steady state or the limit cycle.  To compute the period, we perform an event location procedure.  The procedure first finds a time point when $p_1 = p$ and $\frac{dp_1}{dt}|_{p_1 = p} > 0$, where $p$ is a concentration known to be in the limit cycle.  Then, the algorithm finds the next time point in which $p_1 = p$ and $\frac{dp_1}{dt}|_{p_1 = p} > 0$ and saves this time point.  The period is then taken to be the difference between the two time points.

Figure \ref{fig:per_hill_diff} shows the periods of the two models with respect to the Hill coefficient.  Again, we see that the Hopf bifurcation of model \eqref{eq:model1} ($h \approx 2$, Figure~\ref{fig:per_hill_diff}) happens earlier than that of model \eqref{eq:model2} ($h \approx 4$, Figure \ref{fig:per_hill_diff}).  Interestingly, however, the period of model \eqref{eq:model2} increases more rapidly with respect to $h$ and eventually surpasses the period of model \eqref{eq:model1} ($h \approx 4.75$, Figure \ref{fig:per_hill_diff}).  

Next, we fix the Hill coefficient, $h = 4$, and let the transcription rates vary.  Again, we set $k_1 = k_2 = k_3 = k$ and vary $k$ from 1 to 10.  Figure \ref{fig:per_tr_diff} shows, for both models, the periods of the first protein concentration with respect to the transcription rate.  The Hopf bifurcation for model \eqref{eq:model1} ($k \approx 2$, Figure \ref{fig:per_tr_diff}) occurs significantly earlier than that of model \eqref{eq:model2} ($k \approx 4.2$, Figure \ref{fig:per_tr_diff}).  However, for transcription rates after the Hopf bifurcation of model \eqref{eq:model2}, the periods do not differ notably, suggesting that the variation in the periods of the two models is most sensitive to the Hill coefficient. 

\section{Discussion}

This work advances the theoretical study of cyclic gene repression by generalizing the current repressilator models.  First, we permit more transcription-rate functions than the traditional single-step binding function.  We require only that these functions satisfy a few properties that agree with current biological knowledge.  We also broaden the possible degradation terms beyond first-order degradation.  Again, we require only that these functions satisfy certain biological assumptions.  Finally, we assume first-order translation rates but allow them to vary among mRNAs.  

Our new system retains many advantageous qualitative properties of the previous repressilator after these generalizations.  We proved, for instance, that the system with an odd number of genes has a unique steady state, called the central steady state.  We also showed that the system with an odd number of genes converges to the central steady state or to a periodic orbit.  We worked towards a necessary and sufficient condition for when the central steady state is stable and offered a related conjecture.

For the even case, we characterized when the central steady state exists.  We also give a biological criterion for when a steady state is stable.  However, at the level of generality we propose, we cannot prove the same results as M\"uller \textit{et al}.\ regarding the possible number of steady states.  For specific choices of degradation-rate and transcription-rate functions, one can, however, analyze the limiting dynamics of system \eqref{eq:repress} with $n$ even by using the Poincar\'e-Bendixson Theorem for monotone systems given in \cite{Mallet-Paret1990}.  

Next, we derived new transcription-rate functions from the successive-binding binding assumption.  Recall that the successive-binding function was derived from biological assumptions that are more reasonable than those of the commonly used single-step binding assumption.  In Section 4, we showed that allowing for more general functions can lead to significant changes in dynamics.  For example, numerical simulations showed that amplitudes and periods of a model constructed with the old transcription-rate function and one with our new function differed significantly.  Numerical simulations revealed that the period was most sensitive to the Hill coefficient.  

Going forward, we aim to determine how well the generalizations presented in this work generate more accurate representations of the repressilator.  Specifically, we aim to build off the work of Khammash and Lillacci in \cite{Khammash2010} to compare parameter estimates of previous repressilator models with our generalized model.  The recovered parameters will shed light on certain biological information.  For example, the Hill coefficients in the transcription-rate functions correspond to the number of binding sites on a promoter region.  Next, these fits can shed light on the effectiveness of various transcription-rate and degradation-rate functions.  Finally, we plan to apply the model selection approach from \cite{Khammash2010} to select among hypothesized repressilator models given actual repressilator data. 

In summary, we now better understand stability and limiting dynamics of the repressilator system for a wide range of biologically relevant degradation-rate and transcription-rate functions.  We hope that our results will encourage theoretical and experimental biologists to broaden the possible degradation-rate and transcription-rate functions used to model the repressilator and other gene regulatory networks.  Finally, we expect that allowing general functions for these terms will generate more accurate and predictive models of not only the repressilator but genetic repression in general.  

\subsection*{Acknowledgements}
The authors thank Jake A. Pitt, Ruben Perez-Carrasco, and 2 conscientious referees for their helpful comments and suggestions that helped us improve the work.  AS thanks Mariano Beguerisse D{\'i}az and Heather A. Harrington for helpful discussions.  AS was partially supported by the NSF (DMS-1312473/1513364 and DMS-1752672) and the
Simons Foundation (\#521874).

\bibliographystyle{unsrt}
\bibliography{references}
\end{document}